\documentclass[aps,pra,twocolumn,showpacs,preprintnumbers,nofootinbib]{revtex4-1}

\usepackage{amsmath}        
\usepackage{amsthm}
\usepackage{amsfonts}
\usepackage{hyperref}
\usepackage{graphicx}

\newcommand{\nn}{\nonumber}

\newtheorem{theorem}{Theorem}
\newtheorem{lemma}{Lemma}

\begin{document}

\title{Continuous error correction for Ising anyons}
\author{Adrian Hutter and James R. Wootton}
\affiliation{Department of Physics, University of Basel, Klingelbergstrasse 82, CH-4056 Basel, Switzerland}

\date{\today}

\begin{abstract}
 Quantum gates in topological quantum computation are performed by braiding non-Abelian anyons. These braiding processes can presumably be performed with very low error rates. 
 However, to make a topological quantum computation architecture truly scalable, even rare errors need to be corrected.
 Error correction for non-Abelian anyons is complicated by the fact that it needs to be performed on a continuous basis and further errors may occur while we are correcting existing ones.
 Here, we provide the first study of this problem and prove its feasibility, establishing non-Abelian anyons as a viable platform for scalable quantum computation.
 We thereby focus on Ising anyons as the most prominent example of non-Abelian anyons and show that for these a finite error rate can indeed be corrected continuously.
 There is a threshold error rate $p_c>0$ such that for all error rates $p<p_c$ the probability of a logical error per time-step can be made exponentially small in the distance of a logical qubit.
\end{abstract}

\maketitle

\section{Introduction}

Besides revealing spectacular features of quantum physics, non-Abelian anyons are sought for their potential application in topological quantum computing \cite{kitaev2003,freedman_tqc,nayak_rev,pachos,stern}.
Ising anyons are the most well-studied non-Abelian anyon model, since they describe the exchange statistics of localized Majorana fermions \cite{pachos} and are expected as elementary excitations of the $\nu=\frac{5}{2}$ fractional quantum Hall state \cite{moore,wen}.
A variety of condensed-matter systems have been proposed as potential hosts of Majorana fermions, see Refs.~\cite{alicea_rev,leijnse,beenakker,dassarma} for reviews. 
Ising anyons also appear as excitations \cite{kitaev2006,levin,kapit,palumbo} or ends of defect lines \cite{bombin_ising,you,petrova,wootton2015} in several spin-lattice models.

The set of quantum gates that can be performed topologically, i.e.\ by braiding anyons, depends on how qubits are encoded into the fusion space of a number of Ising anyons.
However, the gate set will not allow for universal quantum computation for any encoding \cite{ahlbrecht}.
In order to perform a universal quantum computation by use of Ising anyons, some gates need to be performed in a non-topological way \cite{bravyi,freedman,bonderson_prl,bonderson}. 
Assuming that all topological operations are error-free, these non-topological operations have a very high error threshold \cite{bravyi}.

Here, we want to focus on the assumption of error-free topological operations. It is often said that topological operations are ``inherently fault-tolerant''. 
However, even a mass gap which is significantly higher than temperature will still lead to a finite density of accidental quasi-particle excitations, and these need to be corrected if a \emph{scalable} quantum computation architecture is to be built. 
The field of error correction for non-Abelian anyons is still relatively young.
Error correction algorithms for Ising anyons \cite{brell} and other non-Abelian anyon models \cite{wootton,hutter} have been benchmarked using Monte Carlo simulations.
Ref.~\cite{burton} demonstrated that even error correction for Fibonacci anyons can be simulated on a classical computer, despite them being universal for quantum computation.
Ref.~\cite{wootton_proof} provides a threshold proof for arbitrary anyon models and a wide class of decoders.

These references assume that we are able to detect all anyonic charges at some time, and then are able to fuse as many anyons as we like, without any further errors occurring.
This picture, however, is highly idealized. 
In reality, further errors may occur while we are correcting existing ones. 
This is both due to the finite time needed to move existing anyons towards each other, and due to the fact that we will in general need more than one round of fusion to get rid of all existing non-Abelian anyons.

For Abelian anyons models (such as the toric code \cite{dennis}), it is possible to record all measurements for some time and only correct a net error at the final time step. 
This is a possibility we do not have with non-Abelian anyons: they need to be corrected ``on the fly'', as the results of their non-Abelian braiding would be impossible to unwind later on. 
Ref.~\cite{pedrocchi} recently pointed out that even performing error correction after the completion of each non-Abelian braid is not sufficient, since the braiding procedure will turn local errors into non-local ones.

This article thus investigates the thus far unexplored problem of \emph{continuous error correction for non-Abelian models}, where we focus on Ising anyons as the most prominent example of non-Abelian anyon. 
We restrict our study to the most trivial topological gate, the identity -- i.e., on the task of preserving a topologically stored quantum state. 
It is generally assumed that the thresholds for quantum information processing are identical to those for quantum information storage.
Our main result is that a sufficiently low rate of errors can indeed be corrected continuously, allowing in principle to preserve a topologically stored quantum state indefinitely in a sufficiently large system.

Sec.~\ref{sec2} discusses continuous error correction for Ising
anyons and states our main result, the proof of which can
be found in Sec.~\ref{sec3}.

\section{Continuous error correction for Ising anyons}\label{sec2}

Following Ref.~\cite{brell}, we consider a square lattice of size $L\times L$ with periodic boundary conditions which hosts \emph{Ising anyons}. 
Anyons are quasi-particles that live in 2D space, and this space needs to be discretized to allow for measurements of the local anyonic charge.
Our setting is thus the most generic choice that is agnostic with respect to the physical realization of the Ising anyons.
The chosen setting is also of interest from a coding-theoretic perspective, in that it provides an example of correcting a code without an explicit Pauli-matrix tensor product structure \cite{brell}.
For proposals to realize Majorana fermions in nanowire hybrid systems, which are experimentally most advanced \cite{mourik,das,deng,rokhinson,churchill,finck,nadj,pawlak}, other geometries may be more adequate \cite{alicea_natphys,pedrocchi}.

Each cell of the square lattice can carry either \emph{vacuum} $1$, a \emph{fermion} $\psi$, or a (non-Abelian) \emph{anyon} $\sigma$. These satisfy the 
fusion rules
\begin{align}\label{eq:fusion}
 \psi\times\psi=1\,,\quad\psi\times\sigma=\sigma\,,\quad\sigma\times\sigma=1+\psi\,.
\end{align}
All $L\times L$ charges are measured periodically at times $0, 1, 2, \ldots$. 
We assume that these measurements can be performed flawlessly.
This assumption means that we neglect the inevitable errors that arise during the complex measurements of non-Abelian anyons.
However, it is common to study this ideal case in quantum error correction to establish a foundation for later studies that include measurement noise.

We study the question of whether it is possible to preserve a quantum state stored in this system despite a constant rate of errors affecting it.
More specifically, we consider whether it is possible to preserve a certain state of the degenerate vacuum of the system. 
Transitions between different vacuum states can be induced by dragging fermions or anyons around the torus.
More realistically, one would consider storing a quantum state in the fusion space of a set of well-separated anyons \cite{brell}.
However, we focus on the task of preserving a certain vacuum state of the torus for simplicity, since at sufficiently large length-scales, the question of correctability is independent of the particular encoding scheme chosen.

Let us assume that between any two rounds of charge measurement a pair of fermions and a pair of anyons are created on each pair of adjacent cells of the lattice with probability $p$ each ($2p<1$). 
Error events can thus be associated with edges of the square lattice.
Note that the case of both a pair of anyons and a pair of fermions being created on the same edge is indistinguishable from only a pair of anyons being created.
We thus restrain from considering this case explicitly.
Hopping and braiding of anyons can emerge from pair creation processes.
While in reality different processes will have different rates, we can think of $p$ as the probability associated with the highest-rate process.
Finally, we assume that we are able to move an anyon or a fermion to an adjacent cell over the course of one measurement period. 

Our goal is to show that for sufficiently low rates $p$, it is possible to perform error correction (for any number of time-steps) such that the probability of failure \emph{per error correction period} is exponentially small in $L$.
Consequently, a given state of the anyonic vacuum can be preserved for a time that is exponentially long in $L$.

The basic idea behind our error correction approach is that 
it is always possible to first fuse all $\sigma$ anyons in pairs, and then fuse all $\psi$ fermions in pairs in order to obtain a vacuum state \cite{brell}.
Note that according to the second of the fusion rules in Eq.~(\ref{eq:fusion}), if a fermion and an anyon move to the same cell, they will fuse to an anyon.
If that anyon is further moved around, it will carry with it the additional fermionic charge. 
We will refer to this informally as the anyon ``swallowing'' the fermion.
The fermion will be recovered if the anyon is brought to fusion with an other anyon, as the fermionic parity is conserved.
We will thus continuously fuse the anyons in pairs, recovering any ``swallowed'' fermions.
Since the fermions are Abelian, error correction for them can be postponted: it suffices to correct a net error after some time.
This is similar to the idea of updating a Pauli frame instead of performing actual corrections in a surface code (see e.g.\ Ref.~\cite{jones}).
If we are moving two anyons towards each other in an attempt to fuse them, further errors may happen along their path that make one of them disappear.
We will shortly discuss how we deal with this.

In order to formally discuss error correction, we consider a $2+1$-dimensional cubic lattice, in which time flows ``upwards'', and with periodic boundary conditions in horizontal (spatial) direction.
Charge measurements correspond to horizontal faces. 
It will prove convenient to identify error events and paths along which we move anyons with edges of the \emph{dual} lattice of this cubic lattice.
Error events happen between consecutive rounds of charge measurement and affect two adjacent cells. They can thus naturally be identified with vertical faces of the primal lattice, and, in turn, horizontal edges of the dual lattice.
We will call the one-cell-per-time-step paths along which we move the anyons during error correction their \emph{world-lines}. 
Error events will not be counted as part of this world-line.
Horizontal faces of the primal lattice (charge measurements) are naturally identified with vertical edges of the dual lattice.
For every charge measurement which detects an anyon, we consider the associated vertical edge of the dual lattice to be part of the anyon's world-line.
If an anyon is moved to an adjacent cell, the horizontal edge connecting the old and new vertical edge is also considered to be part of the anyon's world-line.
Note that error events always correspond to horizontal edges of the dual lattice, while anyon world-lines include both horizontal (intentional movements) and vertical (charge measurements) edges.
Let $W$ denote the union of all world-lines.

Let $A$ denote the set of all anyonic (as opposed to fermionic) error events and $F$ the set of all fermionic error events. 
Both of these are subsets of the horizontal edges of the dual of the $2+1$-dimensional cubic lattice.
Let $\partial A$ denote the set of cells of the cubic lattice which have an \emph{odd} number of elements of $A$ incident upon them.
Elements in $\partial A$ correspond to \emph{unexpected changes in the anyonic charge}, i.e., those which are not due to us intentionally moving an anyon to an adjacent cell.
The sets $W$ and $\partial A$ are known to us with certainty. 

Fig.~\ref{table} summarizes the natural language definition and the geometrical interpretation of all symbols that are relevant for our proof.

\begin{figure*}

\begin{tabular}{|c|c|c|}
\hline
symbol & natural language definition & geometrical interpretation \\
\hline\hline
$W$ & anyon world-lines & horizontal and vertical edges of the dual lattice \\
$A$ & anyonic (as opposed to fermionic) errors &  horizontal edges of the dual lattice \\
$F$ & fermionic errors & horizontal edges of the dual lattice \\
$\partial A$ & unexpected changes in the anyonic charge & cells of the lattice \\
$A_t$ & anyon errors in time-step $t$ & horizontal edges of the dual lattice \\
$\partial A_t$ & unexpected changes in the anyonic charge in time-step $t$ & cells of the lattice \\
$H_t$ & hypothesis about anyonic errors in time-step $t$ &  horizontal edges of the dual lattice \\
$H$ & all hypotheses $H_t$ up to the present time &  horizontal edges of the dual lattice \\
$A^s$ & strings of anyon errors & horizontal edges of the dual lattice \\
$A^l$ & loops of anyon errors & horizontal edges of the dual lattice \\
$A^s_i$ & anyon error strings in connected component $i$ & horizontal edges of the dual lattice \\
$H_i$ & anyon error hypothesis in connected component $i$ & horizontal edges of the dual lattice \\
$W_i$ & anyon world-lines in connected component $i$ & horizontal and vertical edges of the dual lattice \\
$W_i^h$ & intentional anyon movements (horizontal elements of $W_i$) & horizontal edges of the dual lattice \\
$W_i^v$ & charge measurements detecting an anyon (vertical elements of $W_i$) & vertical edges of the dual lattice \\
$P$ & path looping around the torus & horizontal and vertical edges of the dual lattice \\
$O_i$ & loop formed by the disjoint union of $A_i^s$ and $W_i$ & horizontal and vertical edges of the dual lattice \\
$O$ & union of all loops $O_i$ & horizontal and vertical edges of the dual lattice \\
$D_i(P)$ & path $P$ deformed by taking the symmetric difference with loop $O_i$ & horizontal and vertical edges of the dual lattice \\
$P'$ & deformation of $P$ maximizing the number of $A$ events & horizontal and vertical edges of the dual lattice \\
\hline
\end{tabular}
\caption{Natural language definition and the geometrical interpretation of all symbols that are relevant for the proof.}
\label{table}
\end{figure*}

Let $A_t$ and $\partial A_t$ denote the subsets of $A$ and $\partial A$, respectively, that happen between charge measurements $t-1$ and $t$.
Since the fusion rules in Eq.~(\ref{eq:fusion}) preserve the parity of the number of anyons that exist at any given time, the sets $\partial A_t$ always have even cardinality. 
It is the task of a classical error correction algorithm to form a hypothesis about the set $A_t$ that is compatible with the given set $\partial A_t$.
This problem is exactly identical to the well-studied problem of finding the most likely error set in a toric code with bit-flip errors and perfect syndrome measurements \cite{dennis}.
We can thus employ the standard algorithm used to find such a pairing in the toric code case, namely an efficient minimum-weight perfect matching (MWPM) algorithm \cite{edmonds,austin}.
The weight of a path connecting two lattice cells is thereby given by the $2$-dimensional $L_1$-norm, i.e., the Manhattan distance in the $L\times L$ lattice with periodic boundary conditions.
We stress that despite the fact that we deal with anyon world-lines in $2+1$ dimensions, the algorithmic part of the error correction problem for the anyons is a $2$-dimensional one. This is in contrast to the error correction problem for the fermions.

Let us call a subset of the edges of the dual lattice a \emph{string} if there are exactly two cells of the cubic lattice which have exactly one of the edges incident upon them, and all other cells have either zero or two edges incident upon them. 
The MWPM algorithm will return strings that connect the elements in $\partial A_t$ in pairs.
The union of these strings, which we call $H_t$, forms our hypothesis about what anyonic errors have happened between charge measurements $t-1$ and $t$.
Let $H=\bigcup_tH_t$ denote the union of all hypotheses up to the present time.
We note that $A_t$ can in general not be decomposed into strings ending at elements of $\partial A_t$ -- it can contain \emph{loops}. 
These are defined as sets of edges of the dual lattice such that each cell of the cubic lattice has zero or two edges of the set incident upon them. 

Each element in $\partial A$ is connected by a string which is a subset of $H$ to a contemporaneous element of $\partial A$.
Furthermore, each element in $\partial A$ is the beginning or the end of an anyon world-line.
This world-line connects the element either to another element in $\partial A$, which has a different time-coordinate, or to a currently existing anyon.
Each currently existing anyon is thus connected through a chain consisting of strings which are alternately subsets of $W$ and $H$ to another currently existing anyon.
If two currently existing anyons are connected this way, we move them towards each other, one cell per time-step, along the shortest possible path which is homologically equivalent to the chain that connects them \cite{create_comment}.

The error correction problem for the fermions is much more involved than for the anyons
because fermions are not only created and moved by elements of $F$, but also by elements of $A$ and $W$.
Fig.~\ref{fig2} illustrates a process in which two fermions are created while correcting two pairs of anyons, and a process in which a fermion is ``swallowed'' during anyon error correction.
Both of these processes illustrate that it is possible to create pairs of fermions which do not appear in the same measurement period.
For this reason, the error correction problem for the fermions is $2+1$-dimensional.
We need to pair unexpected changes of the fermionic charge which may have different time-coordinates.
(Since before the final time-step we never attempt to move fermions, any change in fermionic charge is unexpected.)
Unexpected appearances of fermions are due to fermion error events, or due to fusion of two anyons. 
Unexpected disappearances are due to an anyon ``swallowing'' a fermion, or due to an appearance event at a location where a fermion has already been present.
When applying MWPM to the fermionic problem, the weight we assign to connecting two spatio-temporally separated events is the $2+1$-dimensional $L_1$-norm.
Clearly, a more sophisticated weight would take knowledge about anyon world-lines into account.
This would be similar to the idea of introducing ``shortcuts'' in Ref.~\cite{hutter}.
It may also help to weight spatial and temporal distances differently, and to take entropic contributions to the weight into account.
However, we restrict to the $L_1$-norm for simplicity.

\begin{figure}
\centering
\includegraphics[width=1.0\columnwidth]{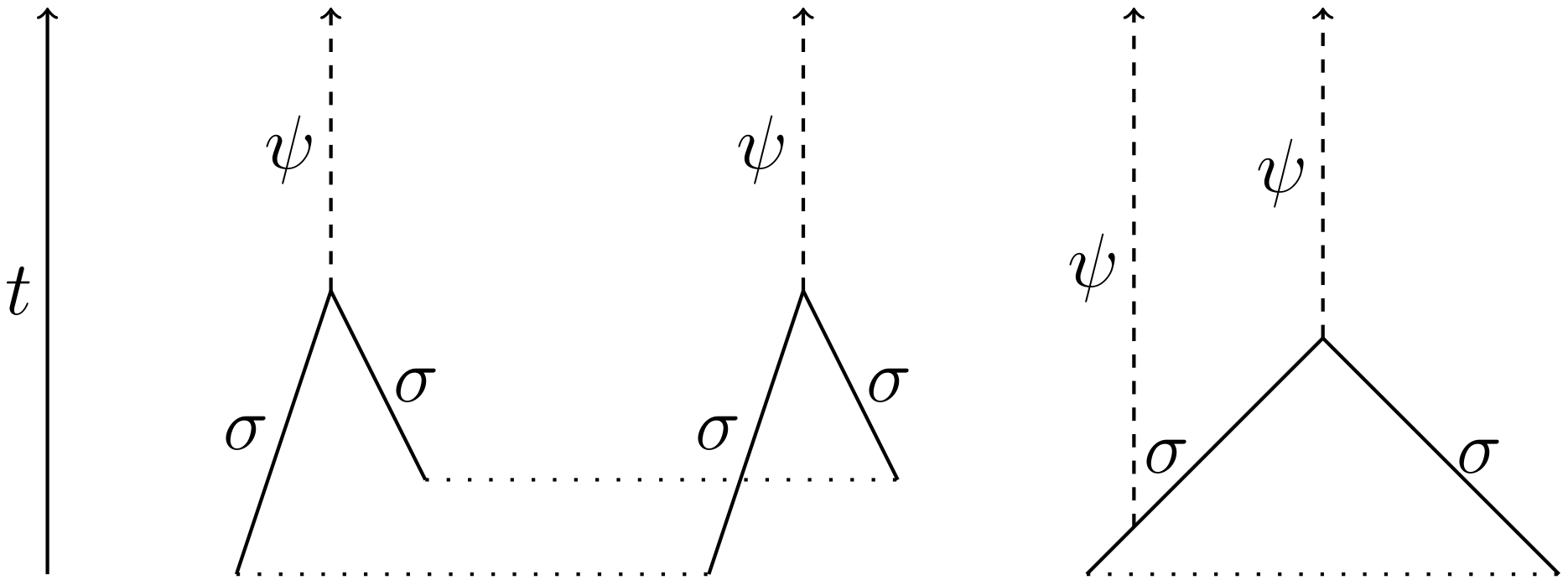}
\caption{Two possible processes illustrating how creation and fusion of $\sigma$ anyons can produce or ``swallow'' $\psi$ fermions. Anyonic errors are dotted, anyonic world-lines are solid, and fermionic world-lines are dashed. 
Left process: Two error strings produce two pairs of anyons, which are incorrectly paired and correspondingly brought to fusion. This process creates a pair of fermions with probability $\frac{1}{2}$.
Right process: An error string creates a pair of anyons which are brought to fusion. Along one of the anyonic world-lines, a pair of fermions is created and one of the two fermions is ``swallowed'' by the nearby anyon.
The second fermion is recovered only when the two anyons are fused.}
\label{fig2} 
\end{figure}

Our main result is the following theorem.

\begin{theorem}\label{thm}
If error correction is performed as described above, there is a finite threshold $p_c>0$ such that for $p<p_c$ the error rate per error correction period on the stored quantum information is exponentially small in $L$.
\end{theorem}

While we prove the existence of a finite lower bound for $p_c$, this lower bound is very small.
Nevertheless, it demonstrates the important fact that the threshold is non-zero, and hence that continuous error correction for non-Abelian anyons is possible in principle.
We expect the value for the lower bound to be very pessimistic, and so it should not be confused for an estimated value for $p_c$.

\section{Proof of the main result}\label{sec3}

Our proof is similar in nature to the proofs for the correctability of the toric code by means of MWPM \cite{dennis,fowler}.
As long as we only consider the (non-Abelian) anyons, the error model for them is exactly identical to the one for the toric code \cite{kitaev2003,dennis} with bit-flip rate $p$ and perfect syndrome measurements. 
The correctabiliy of the anyons thus follows from the correctability of this (very well-studied) error model.
Indeed, Ref.~\cite{dennis} contains an analytical proof that for this problem $p_c\geq3.7\%$.
 
Our main difficulty is correcting the fermions which may be produced or ``swallowed'' during the continuous correction of the anyons, leading to a more involved, correlated effective error model for the fermions.
Fig.~\ref{fig2} shows two examples of such processes.
From here on, we consider the \emph{hypothetical completion} of $W$. 
That is, we consider the \emph{hypothetical} world-lines $W$ that we would obtain if we could complete error correction in accordance with our hypothesis $H$ at a given time and bring all anyons to fusion, without any further errors occurring.
There is thus no longer a notion of ``currently existing anyons''. Each string in $W$ (anyon world-line) begins and ends at an element of $\partial A$ (unexpected change in anyonic charge).

The hypothetical completion of $W$ is introduced in order to study whether error correction has been successful up to a given time, and is used to avoid explicitly modeling a realistic fault-tolerant read-out step.
We stress that no assumption is made that we can actually complete error correction without any new errors occurring in reality.
Similarly, we assume that the system is initially free of anyonic defects to avoid explicitly modeling a fault-tolerant read-in step.
The idealized assumptions of an error-free initial state and a final fault-less period of error correction are standard in the study of fault-tolerant qubit-based quantum computation (see, e.g.\ Ref.~\cite{groszkowski}).
While read-in and read-out will have to be performed in a fault-tolerant way in reality, we do not consider these for simplicity and in order to keep our results independent from the particular encoding scheme.

We have remarked that the set $A$ can be decomposed into loops and strings which connect elements in $\partial A$ in pairs.
We choose this decomposition such that each element in $\partial A$ has exactly one string incident upon it.
Let $A=A^s\cup A^l$ be such a decomposition.  (Note that the decomposition is in general not unique.)
The sets $A^s$, $H$, and $W$ can then all be decomposed into strings, each of which ends at an element of $\partial A$. 
Conversely, each element of $\partial A$ has three strings incident upon it, which are respectively subsets of  $A^s$, $H$, and $W$.

Note that the sets $H$ and $W$ are not necessarily disjoint: it can happen that we attempt to move an anyon to an adjacent cell and, between the same two rounds of charge measurement, an anyonic error affecting the same two cells happens.
Let us thus study the disjoint union
\begin{align}\label{eq:graph}
&A^s\sqcup H\sqcup W = \nn\\&\quad \lbrace(e,a) : e\in A^s\rbrace \cup \lbrace(e,h) : e\in H\rbrace \cup \lbrace(e,w) : e\in W\rbrace\,.
\end{align}
Here, the index $i\in\lbrace a, h, w\rbrace$ in the ordered pair $(e,i)$ tells us which of the three sets $A^s$, $H$, or $W$ the edge $e$ belongs to.
If, for example, an edge is an element of both $H$ and $W$, there will thus be two corresponding edges in $A^s\sqcup H\sqcup W$.

The set $A^s\sqcup H\sqcup W$ forms a trivalent graph, with each vertex corresponding to an element of $\partial A$, and having an $A^s$, an $H$-, and a $W$-string incident upon it.
Let us study minimal connected components of this graph. 
Let $A_i^s$, $H_i$, and $W_i$  denote the union of all strings in  $A^s$, $H$, and $W$, respectively, that belong to connected component $i$.
Finally, let $W_i^h$ denote the set of horizontal edges 
(i.e., intentional anyon movements) 
and $W_i^v$ the set of vertical edges 
(i.e., charge measurements which detect an anyon) 
in $W_i$.
Recall that $A_i^s$ and $H_i$ consist of horizontal edges only.

From the way our error correction procedure is defined, we have the inequalities
\begin{align}\label{eq:inequalities}
 |W_i^h| \leq |H_i| \leq |A_i^s|\,.
\end{align}
The first inequality is due to us moving anyons along the shortest path which is homologically equivalent with $H_i$. 
We could always choose $|W_i^h|=|H_i|$ by undoing exactly the errors that happened according to our hypothesis.
The second inequality is due to using MWPM for error correction.
Assume by contradiction that $|H_i|>|A_i^s|$. Then, replacing $H_i$ with $A_i^s$ would yield a perfect matching of the unexpected changes in anyonic charge which is of lower weight than the one returned by the MWPM algorithm, which contradicts its definition.

Now consider the loops $O_i=A_i^s\sqcup W_i$.
The following lemma is the main technical tool that we use in order to deal with these.

\begin{lemma}\label{lem1}
All loops $O_i$ satisfy
 \begin{align}
  |O_i|\leq 4|A_i^s|\,.
 \end{align}
\end{lemma}

Note that for the simplest possible process, a single anyon error event that is immediately corrected ($|A_i^s|=1$, $|H_i|=1$, $|W_i^h|=1$, $|W_i^v|=2$, $|O_i|=4$), the bound is tight.

\begin{proof}
Let $f_i$ denote the number of fusion events of a pair of anyons in $O_i$.
An anyon needs to be moved away from each location at which it appears. However, before fusion two anyons may be adjacent so that we need to move only one of them. We thus have
\begin{align}\label{eq:move}
 |W_i^h|\geq |W_i^v|-f_i\,.
\end{align}
Furthermore, each error event can create at most two anyons, so 
\begin{align}\label{eq:create}
 f_i\leq |A_i^s|\,.
\end{align}
Recall from Eq.~(\ref{eq:inequalities}) that
\begin{align}\label{eq:inequality}
 |W_i^h|\leq|A_i^s|\,.
\end{align}
Combining the above inequalities, we find
\begin{align}\label{eq:combined}
 |W_i^v|\leq 2|A_i^s|\,.
\end{align}
For the total length of the loop, we find, using Eqs.~(\ref{eq:inequality}) and (\ref{eq:combined}),
\begin{align}
 |O_i| = |A_i^s|+|W_i^h|+|W_i^v| \leq 4|A_i^s|\,.
\end{align}
\end{proof}

The following lemma provides a necessary condition for the failure of error correction.

\begin{lemma}\label{lem2}
 A failure of error correction requires a homologically non-trivial closed path $P$ (a loop) satisfying 
\begin{align}\label{eq:cond}
 7|P\cap A| + |P\cap F| \geq |P|/2\,.
\end{align}
\end{lemma}

\begin{proof}
Let us first study the possibilities for error correction failing for the anyons (as opposed to the fermions).
Recall that we have decomposed the set of anyonic errors $A$ into loops $A^l$ and strings $A^s$.
If one of the loops which are subsets of $A^l$ is homologically non-trivial, Eq.~(\ref{eq:cond}) will obviously be satisfied, as we can choose $P$ to be the corresponding loop and have $|P\cap A|=|P|$.
The second possibility for error correction for the anyons failing is that one of the loops $O_i$ is topologically non-trivial. In this case, we choose $P=O_i$ and are done, since by use of Lemma~\ref{lem1} we have
  \begin{align}
   7|O_i\cap A| = 7|A_i^s| \geq \frac{7}{4}|O_i|\,.
  \end{align}
So assume from now on that all loops which are subsets of $A^l$ and all loops $O_i$ are homologically trivial, and that error correction failing is due to the fermionic part of the problem.
 
 Clearly, MWPM failing to correct the fermions requires that there be a homologically non-trivial closed path $P$ containing at least $|P|/2$ edges that have been affected by an event that can possibly have created or moved fermions, 
 for otherwise the minimum-weight correction of the fermions will never move a fermion around the torus.
We assume pessimistically that each edge in $A^l$ and in $O=\bigcup_iO_i$ (anyon error event or anyon world-line) counts as a potential fermion error event. 
 So formally, we need a path $P$ with
 \begin{align}
  |P\cap(A^l\cup O\cup F)| \geq |P|/2\,.
 \end{align}
 We will prove that if there is such a path $P$, there is a (possibly identical) path $P'$ which is homologically equivalent to $P$ and satisfies the inequality in the lemma, i.e., $7|P'\cap A| + |P'\cap F| \geq |P'|/2$. 

 Given a loop $O_i$ with $O_i\cap P\neq\emptyset$, we can consider the ``deformed'' path $D_i(P)=(P\setminus O_i)\cup(O_i\setminus P)$. 
 The path $P'$ is obtained by applying a (possibly empty) set of deformation operations $D_i$ to $P$. Since all of the loops $O_i$ are homologically trivial, the deformed path $P'$ will be homologically equivalent to $P$.
 We define the path $P'$ such that the number of $A$ events in the path is maximized; i.e., such that 
 \begin{align}\label{eq:PprimeDefi}
  |P'\cap A_i^s|=\max\lbrace|P\cap A_i^s|, |D_i(P)\cap A_i^s|\rbrace\,,
 \end{align}
 for all loops with $O_i\cap P\neq\emptyset$. 
Equivalently, the path $P'$ is defined such that
\begin{align}\label{eq:consequence}
  |A_i^s\setminus P'|\leq|A_i^s\cap P'|\,.
\end{align}

 Let us define $\tilde{A}^l=A^l\setminus O$ and $\tilde{F}=F\setminus O$.
 By assumption, we have
 \begin{align}
 0 &\leq |P\cap(A^l\cup O\cup F)| - |P|/2 \nn\\
&= |P\cap \tilde{A}^l| + |P\cap O| + |P\cap \tilde{F}| - |P|/2 \nn\\
&= |P\cap \tilde{A}^l| + |P\cap O| + |P\cap \tilde{F}| \nn\\&\quad - (|P\cap O|/2 + |P\setminus O|/2)\,.
 \end{align}
Note that $P\setminus O$ is not affected by deformation operations, i.e. $P\setminus O = P'\setminus O$. 
Since $P\cap \tilde{F}\subseteq P\setminus O$, we also have $P\cap \tilde{F} = P'\cap \tilde{F}$ and similarly $P\cap \tilde{A}^l = P'\cap \tilde{A}^l$.
Therefore
\begin{align}\label{eq:estimate}
 0 &\leq |P'\cap \tilde{A}^l| + |P\cap O|/2 + |P'\cap \tilde{F}| - |P'\setminus O|/2 \nn\\
&= |P'\cap \tilde{A}^l| + (|P\cap O| + |P'\cap O|)/2 + |P'\cap \tilde{F}| \nn\\&\quad - (|P'\setminus O|/2 + |P'\cap O|/2) \nn\\ 
&= |P'\cap \tilde{A}^l| + \sum_i(|P\cap O_i| + |P'\cap O_i|)/2  \nn\\&\quad + |P'\cap \tilde{F}| - |P'|/2\,. 
\end{align}
 If $\underline{P\cap O_i=P'\cap O_i}$ we find
 \begin{align}\label{eq:xxx}
  &(|P\cap O_i| + |P'\cap O_i|)/2 \nn\\&\quad = |P\cap O_i| \nn\\&\quad = |P\cap A_i^s| + |P\cap O_i\setminus A_i^s|\,.
 \end{align}
Since $P\cap O_i\setminus A_i^s\subseteq O_i\setminus A_i^s$ and $A_i^s\subseteq O_i$ we have
\begin{align}\label{eq:yyy}
 |P\cap O_i\setminus A_i^s| \leq |O_i\setminus A_i^s| = |O_i|-|A_i^s|\,.
\end{align}
Combining Eqs.~(\ref{eq:xxx}) and (\ref{eq:yyy}) with Lemma~\ref{lem1}, we arrive at
\begin{align}
 (|P\cap O_i| + |P'\cap O_i|)/2 \leq |P\cap A_i^s| + 3|A_i^s|\,.
\end{align}
Using Eq.~(\ref{eq:consequence}), we obtain
 \begin{align}
  |A_i^s| = |A_i^s \cap P'| + |A_i^s \setminus P'| \leq 2|A_i^s \cap P'| = 2|A_i^s \cap P|\,.
 \end{align}
 We finally find
 \begin{align}\label{eq:final}
  (|P\cap O_i| + |P'\cap O_i|)/2 \leq 7|P'\cap A_i^s|\,.
 \end{align}
If, on the other hand, $\underline{P\cap O_i\neq P'\cap O_i}$, we find, using Lemma~\ref{lem1} for the first inequality, 
\begin{align}
 (|P\cap O_i| + |P'\cap O|_i)/2 = |O_i| /2 \leq 2|A_i^s| \leq 4|P'\cup A_i^s| \,.
\end{align}
 So in both cases Eq.~(\ref{eq:final}) holds and we find from Eq.~(\ref{eq:estimate}) that 
\begin{align}
 0 &\leq |P'\cap \tilde{A}^l| + 7\sum_i|P'\cap A_i^s| + |P'\cap \tilde{F}| - |P'|/2 \nn\\
&\leq 7|P'\cap A| + |P'\cap F| - |P'|/2\,.
\end{align}
\end{proof}

Theorem~\ref{thm} directly follows from the following lemma.

\begin{lemma}
 The probability per time-step of a path as in Lemma~\ref{lem2} is exponentially suppressed with $L$ if $p<15^{-14}\approx 3\times10^{-17}$.
\end{lemma}
\begin{proof}
 Consider two lines of length $L$ looping in homologically non-equivalent ways around the torus. Path $P$ needs to cross at least one of them.
 Since the two lines can be crossed at $O(L)$ locations, and a path in a three-dimensional cubic lattice can at each step turn into $5$ directions,
 there are at most $5^{\ell+O(\log(L))}$ closed paths of length $\ell$ in the lattice crossing any of the two lines at a given time. 
 Let $n=|P\cap A|+|P\cap F|$ be the number of error events along the path, and let $\ell=|P|$.
 For a path satisfying $7|P\cap A|+|P\cap F|\geq |P|/2$, we need $14n\geq\ell$.
 With fixed locations of the $n$ errors, the probability of such a path is at most $p^n\leq p^{\ell/14}$.
 In a path of length $\ell$, there are no more than $3^\ell$ possibilities for picking the locations of $A$ and $F$ events.
 The probability per time-step of a path satisfying $7|P\cap A|+|P\cap F|\geq |P|/2$ is thus upper-bounded by 
\begin{align}
 \sum_{\ell=L}^\infty5^{\ell+O(\log(L))}3^\ell p^{\ell/14}\,,
\end{align}
 which is exponentially suppressed with $L$ if $15p^{1/14}<1$.
\end{proof}

\section{Conclusions}

Topological quantum computing holds the promise of processing quantum information with hardware which has intrinsically low error rates.
Still, these rare errors need to be corrected in a truly large-scale computation.
Here, we have demonstrated the feasibility of this task for Ising anyons under the realistic assumption that errors keep happening as we correct those from previous rounds.

Note that another approach towards error correction would be to use qubits encoded in non-Abelian anyons as physical qubits for a further round of error correction. This would lead to a standard decoding problem, although perhaps with a rather complex error model due to the underlying topological processes. In such a case, the available gate set on the final logical qubits would be restricted by the code used in the additional round of error correction. This would remove some of the advantages of using non-Abelian anyons, and may carry a large resource overhead in comparison to standard proposals for physical qubits.

One would hope for a threshold proof for further non-Abelian anyon models, including those for which MWPM cannot be applied to perform error correction \cite{wootton,hutter,burton}.
Unfortunately, the highly general proof in Ref.~\cite{wootton_proof} does not allow for straightforward generalization to the continuous case.
Furthermore, it is an open problem to study error correction for non-Abelian anyons with charge measurements that can give an incorrect result.

Finally, it would be valuable to get a better idea of the ``true'' threshold for our setup. 
Given the crudeness of our arguments, we expect our threshold of  $p_c\approx3\times10^{-17}$ to be rather pessimistic. 
A better estimate of the true threshold value $p_c$ could be obtained via a more ingenious analytical approach, or by numerical simulations, extending the work of Ref.~\cite{brell} to the continuous case. 
The ``true'' thresholds for Abelian models can often be assessed by finding the phase-transition in a related classical statistical mechanics model \cite{dennis,bombin,andrist}.
Whether something similar can be done for non-Abelian models remains an open problem.

\begin{acknowledgments}
The authors gratefully acknowledge Courtney Brell and Daniel Loss for careful reading of the manuscript and helpful comments.
This work was supported by the SNF and NCCR QSIT.
\end{acknowledgments}


\end{document}